\documentclass[11pt]{article}

\usepackage[english]{babel}
\usepackage{bm}
\usepackage[letterpaper,top=2cm,bottom=2cm,left=3cm,right=3cm,marginparwidth=1.75cm]{geometry}

\usepackage{amsmath}
\usepackage{amssymb}
\usepackage{amsthm}
\usepackage{graphicx}
\usepackage[colorlinks=true, allcolors=blue]{hyperref}

\newtheorem{proposition}{Proposition}[section]

\newtheorem{theorem}{Theorem}[section]
\newtheorem{lemma}[proposition]{Lemma}
\newtheorem{definition}[proposition]{Definition}
\newtheorem{remark}{Remark}

\title{Periodic Response Solutions to Multi-Dimensional Nonlinear Schr\"odinger equation with unbounded perturbation}
\author{Zuhong You \thanks{The work was supported  by National Natural Science Foundation of China ( Grant NO. 12371189).} $^a$ and Xiaoping Yuan \footnotemark[1] $^b$  \\
			$^a$ School of Mathematical Sciences, Fudan University, Shanghai 200433, P. R. China \\
             $^b$ School of Mathematical Sciences, Fudan University, Shanghai 200433, P. R. China }

\date{}
             
\begin{document}
\maketitle
\begin{abstract}
    By applying the Craig-Wayne-Bourgain (CWB) method, we establish the existence of periodic response solutions to multi-dimensional nonlinear Schr\"{o}dinger equations (NLS) with unbounded perturbation.
\end{abstract}
\section{Introduction}

In the last 40 years, significant progress has been made in the KAM (Kolmogorov-Arnold-Moser) theory for nonlinear Hamiltonian partial differential equations (PDEs). The first existence results were established by Kuksin\cite{Kuk87} and Wayne\cite{Way90}, who studied nonlinear wave (NLW) and  Schr\"{o}dinger equations (NLS) in one spatial dimension under Dirichlet boundary conditions.
There are two main approaches: the classical KAM technique (e.g., \cite{Kuk87}, \cite{Way90},  \cite{Kuk93}, \cite{Pos96}, \cite{LY00}, \cite{eliasson2010kam},  \cite{yuan2021kam}) and the Craig-Wayne-Bourgain (CWB) method (e.g.,  \cite{craig1993newton}, \cite{bourgain1994construction},\cite{bourgain1995construction}, \cite{bourgain1998quasi},\cite{bourgain2005green}, \cite{BP11duke}, \cite{berti2013quasi}, \cite{wang2016energy}, \cite{BB20} ). There are too many works to list here in this field.

In the context of unbounded perturbations, where the nonlinearity involves derivatives, KAM theory encounters greater challenges compared to the bounded case. The initial KAM theorem for unbounded perturbations was pioneered by Kuksin\cite{Kuk98kdv}\cite{Kuk00}, who studied small denominator equations with large variable coefficients and developed applicable KAM theorems to investigate the persistence of finite-gap solutions for the KdV equation. The estimate concerning the small-denominator equation with
large variable coefficients is now called the Kuksin lemma\cite{Kuk87}.
Later, Liu and Yuan\cite{LY10spec}\cite{LY11dnls} extended the Kuksin lemma to the limiting case and established the KAM theorem for quantum Duffing oscillators, derivative nonlinear Schrödinger equations (DNLS) and and Benjamin-Ono equations. The Italian school (e.g., \cite{BBP13dnlw}, \cite{BBM14}, \cite{BBM16kdv}, \cite{BBE18}) has developed a novel approach for addressing unbounded perturbations.

All existent approaches dealing with
unbounded perturbations rely on classical KAM techniques. As for the CWB method, it was only briefly mentioned by Bourgain in \cite{bourgain1994construction}, where he suggested in the final remark that his analysis could be extended to unbounded perturbations, such as the Hamiltonian “derivative” wave equation. However, no details were provided, and  there have been no further developments applying the CWB method for unbounded perturbations until today. Moreover, all existing results for unbounded perturbations are restricted to the case where the spatial dimension is equal to $1$. To the best of our knowledge, there are currently no established results for cases where the perturbations are unbounded and the spatial dimension is greater than one.
In this paper, we take a step towards addressing this gap by employing the CWB method to study periodic response solutions of multi-dimensional nonlinear  Schr\"{o}dinger equations with fractional derivative perturbations (unbounded ones).

Consider the multi-dimensional nonlinear Schr\"{o}dinger equations (NLS) with fractional derivatives:
\begin{equation}
    i u_{t}-\Delta u+u+D^{\alpha}(|u|^{2}u)=\varepsilon P(x,\omega t)
    \label{peiodic force equation}
\end{equation}
under periodic boundary conditions (i.e., $x\in \mathbb{T}^{d}$). Here,  
$P(x,\theta):\mathbb{T}^{d}\times \mathbb{T}\to \mathbb{R}$ is Gevrey smooth. Furthermore, the fractional derivative $D^{\alpha}$ is defined as follows:

Let  $f\in L^{2}(\mathbb{T}^{d})$ be a smooth periodic function on the torus $\mathbb{T}^{d}$. We define
\begin{equation}
    D^{\alpha} f(x)=\sum\limits_{n\in \mathbb{Z}^{d}} \langle n\rangle^{\alpha}\hat{f}(n)e^{in\cdot x},
\end{equation}
where $\hat{f}(n)$ is the $n$-th Fourier coefficient of $f$ and  $\langle n\rangle=\sqrt{\sum\limits_{1\le j\le d} |n_{j}|^{2}+1}$.

Furthermore, we define the Gevrey norm of  $P$ as follows:
\begin{equation}
    \lVert P \rVert_{c}=\sum\limits_{n,k} |\hat{P}(n,k)|e^{2(|n|+|k|)^{c}},
\end{equation}
where $\hat{P}(n,k)$ refers to the Fourier coefficients of $P$ and $c$ is referred to the Gevrey index.

We assume $\omega\in[1,2]$ is the parameter. Our aim is to prove that there is a $\frac{2 \pi}{\omega}$-time periodic solution to (\ref{peiodic force equation}) for `` most " $\omega\in [1,2]$. 
More precisely, we have the following theorem:
\begin{theorem}
    Assume $\alpha<\frac{1}{30(d(2d+2)^{d+2}+2)}$,   $P(x,\theta): \mathbb{T}^{d}\times \mathbb{T}\to \mathbb{R}$ is Gevrey smooth, $c$ is small enough and $\lVert P \rVert_{c}<1$. Then, for $\varepsilon>0$ small enough, there exists a set $I_{\varepsilon}\subset [1,2]$ with
    \begin{equation}
        \textup{mes }  I_{\varepsilon} \to 0, \textup{ as } \varepsilon\to 0,
    \end{equation}
    such that for any $\omega\in [1,2]\setminus I_{\varepsilon}$, there exists a Gevrey smooth function 
    \begin{equation}
        u: \mathbb{T}^{d} \times \mathbb{R} \to \mathbb{R},
    \end{equation}
    which solves (\ref{peiodic force equation})  and is $\frac{2\pi }{\omega}$ periodic in time. Moreover, we have
    \begin{equation}
        \sum\limits_{n,k} |\hat{u}(n,k)|e^{\frac{1}{2}(|n|+|k|)^{c}}<C \varepsilon^{1/4},
    \end{equation}
    where $\hat{u}(n,k)$ is the Fourier coefficient of $u$ and $C$ is a constant depending on $d,\alpha$ and $c$.
    \label{maintheorem}
\end{theorem}

\begin{remark}
    In this paper, the Gevrey index $c$ is taken to be small enough for technical simplicity.  To extend the proof to arbitrary values of $c\in (0,1)$, a more refined coupling lemma and more precise estimates of the Green's function are required.  
\end{remark}
\begin{remark}
    The nonlinearity $|u|^{2}u$ can be replaced by any $\frac{\partial H}{\partial \bar{u}}$ where $H(u, \bar{u})=\sum a_{j,l}u^{j}\bar{u}^{l}$, $a_{jl}=a_{lj}\in\mathbb{R}$, is real analytic. When $H$ is a polynomial, the proof follows the same line as described in this paper. When $H$ is real analytic instead of polynomial, combining the scheme in this paper and  (1) of Section 14 in \cite{bourgain1998quasi} gives the proof. 
\end{remark}
\begin{remark}
    The idea to handle the unbounded perturbation is to define non-singular sites to be $(n,k)$ which makes $\min\limits_{\pm}|D_{\pm, n,k}|=\min\limits_{\pm} |\pm k\omega+|n|^{2}+1|>N^{\delta}$ for some small $\delta$, which is different from the usual ones and allows us to handle the unbounded perturbation for $(\alpha<\delta)$.  Only when $\delta$ is small, can the existing separation lemma for $|n|^{2}$  handle it; thus,  $\alpha$ in this paper is small.
\end{remark}

We introduce some useful notations.

For any positive numbers $a, b$ the notation $a\lesssim b$ means $Ca\le b$ for some constant $C>0$. By $a\ll b$ we mean that the constant $C$ is very large.  The various constants can be defined by the context in which they arise. Finally, $N^{a-}$ means $N^{a-\epsilon}$ with some small $\epsilon>0$ (the precise meaning of ``small" can again be derived from the context).

\section{Initial approximate solution}

Choose and fix the various constants $c$, $M$, $C_{1}$, $C_{2}$, $C_{3}$ such that
\begin{equation}
    M>100,\ \ 0<c<\frac{\log \frac{17}{16}}{\log M},\ \ C_{1}=20d,\ \  C_{2}>2,\ \  C_{3}>C_{2}+2.
\end{equation}

Replacing $u$ by $\varepsilon^{\frac{1}{3}} u$, equation (\ref{peiodic force equation}) becomes
\begin{equation}
    iu_{t}-\Delta u+u+\varepsilon^{\frac{2}{3}} D^{\alpha}(|u|^{2} u)=\varepsilon^{\frac{2}{3}} P(x,\omega t).
    \label{scale}
\end{equation}
To solve a periodic solution with frequency $\omega$ of (\ref{scale}) is equivalent to solve
\begin{equation}
    i\partial_{\omega} u-\Delta u+u+\varepsilon^{\frac{2}{3}} D^{\alpha}(|u|^{2} u)=\varepsilon^{\frac{2}{3}} P(x,\theta),
    \label{scale'}
\end{equation}
where $\partial_{\omega} u=\omega\partial_{\theta}$. Let
\[
F(u)=i\partial_{\omega}u-\Delta u+u+\varepsilon^{\frac{2}{3}}D^{\alpha}(|u|^{2}u)-\varepsilon^{\frac{2}{3}}P(x,\theta).
\]
\begin{definition}
    For $f\in L^{2}(\mathbb{T}^{d+1})$, we denote
    \begin{equation}
        \Gamma_{N} f(x,\theta)=\sum\limits_{|n|<N, |k|<N} \hat{f}(n,k) e^{i n\cdot x+ik\theta}. 
    \end{equation}
\end{definition}
\begin{definition}
    We say $u(x,\theta)\in L^{2}(\mathbb{T}^{d+1})$ is an $O(\delta)$-approximate solution to (\ref{scale'}) if 
    \begin{equation}
        \lVert F(u) \rVert\le \delta,
    \end{equation}
    where $\lVert \cdot \rVert$ denote the  norm of $L^{2}(\mathbb{T}^{d+1})$.
\end{definition}

\begin{lemma}
    Let $j_{0}=\lfloor \frac{\log\log \frac{1}{\varepsilon}}{2\log M} \rfloor$. There exists $u_{j_{0}}(x,\theta;\omega)\in L^{2}(\mathbb{T}^{d+1})$ smoothly defined for $\omega\in [1,2]$ such that 
    \begin{itemize}
        \item[\textup{(i)}] $\textup{supp }\hat{u}_{j_{0}}\subset B(0,M^{j_{0}})$.
        \item[\textup{(ii)}] $\lVert u_{j_{0}} \rVert<\varepsilon^{\frac{1}{5}}$.
        \item[\textup{(iii)}] $|\hat{u}_{j_{0}}(\xi)|<e^{-|\xi|^{c}} $ and $|\partial \hat{u}_{j_{0}}(\xi)|<2 e^{-|\xi|^{c}}$, where $\xi=(n,k)\in\mathbb{Z}^{d+1}$ and $\partial$ refers to the derivative with respect to $\omega$. 
    \end{itemize}
    Furthermore, there exists a subset $\Lambda_{j_{0}}\subset [1,2]$ which is a union of disjoint intervals of size $\exp(-j_{0}^{C_{3}})$ such that 
    \[
    \textup{mes }[1,2]\setminus \Lambda_{j_{0}}\lesssim (\log \frac{1}{\varepsilon})^{-1}+e^{-j_{0}-9}.
    \]
    For $\omega\in\Lambda_{j_{0}}$, $u_{j_{0}}$ is an $O(e^{-2(M^{j_{0}})^{c}})$-approximate solution to (\ref{scale'}), i.e., $\lVert F(u_{j_{0}}) \rVert<e^{-2(M^{j_{0}})^{c}}$. Moreover, for  $\omega\in\Lambda_{j_{0}}$,  we also have
    \[
    \lVert \partial F(u_{j_{0}}) \rVert<e^{-2(M^{j_{0}})^{c}}.
    \]
    \label{lemma2.3}
\end{lemma}

\begin{proof}
    Partition $[1,2]$ into intervals $\{I_{\zeta}\}$ of size $\frac{1}{1-e^{-j_{0}-9}}\exp(-j_{0}^{C_{3}})$. Let 
    \begin{align}
         &\Lambda_{j_{0}}\notag\\
         =&\bigcup \{ (1-e^{-j_{0}-9})I_{\zeta}: \exists\  \omega'\in I_{\zeta} \textup{ such that } 
         |-k\omega'+|n|^{2}+1|\ge (\log \frac{1}{\varepsilon})^{-1}(1+|k|)^{-\tau},\notag\\
         &\ \ \  \ \ \      |n|<M^{j_{0}}, |k|<M^{j_{0}}\}.\notag
    \end{align}
    Here, $\tau=d+1$ and $(1-e^{-j_{0}-9})I_{\zeta}$ refers to $(1-e^{-j_{0}-9})$-dilation of $I_{\zeta}$.
    For $\omega\in \Lambda_{j_{0}}$, we have 
    \begin{align}
        |-k\omega+|n|^{2}+1|&>(\log\frac{1}{\varepsilon})^{-1}(1+|k|)^{-\tau}-\exp(-\frac{1}{2}j_{0}^{C_{3}})\notag\\
        &>\frac{1}{2} (\log \frac{1}{\varepsilon})^{-1}(1+|k|)^{-\tau}.\notag
    \end{align}
    Thus, we have
    \begin{align}
        &[1,2]\setminus \Lambda_{j_{0}}\\
        \subset& \mathop{\cup}\limits_{\substack{|n|<M^{j_{0}}\\|k|<M^{j_{0}}}}\{\omega\in[0,1]:|-k\omega+|n|^{2}+1|<(\log \frac{1}{\varepsilon})^{-1}(1+|k|)^{-\tau}\}\cup(\mathop{\cup}\limits_{\zeta} I_{\zeta}\setminus (1-e^{-j_{0}-9})I_{\zeta}).\notag
    \end{align}
    Denote 
    \begin{equation}
        R_{n,k}=\{\omega\in[1,2]: |-k\omega+|n|^{2}+1|<(\log\frac{1}{\varepsilon})^{-1}(1+|k|)^{-\tau}\}.
    \end{equation}
    Note that $R_{n,0}=\emptyset$ and $R_{n,k}=\emptyset$ for $|n|>2|k|$.
    Thus, we have 
    \begin{align}
        \textup{mes } [1,2]\setminus \Lambda_{j_{0}}&< \sum\limits_{\substack{|k|<M^{j_{0}}\\k\ne 0}} \sum\limits_{n\le 2|k|} \textup{mes } R_{n,k}+e^{-j_{0}-9}\notag\\
&\lesssim \sum\limits_{\substack{|k|<M^{j_{0}}\\k\ne 0}} (\log\frac{1}{\varepsilon})^{-1}(1+|k|)^{-d-1}\cdot |k|^{d-1}+e^{-j_{0}-9}\notag\\
&\lesssim (\log \frac{1}{\varepsilon})^{-1}+e^{-j_{0}-9}.\notag
    \end{align}
For $\omega\in\Lambda_{j_{0}}$, let 
\[
\hat{u}_{j_{0}}(n,k)=
\begin{cases}
\frac{\varepsilon^{\frac{2}{3}}\hat{P}(n,k)}{-k\omega+|n|^{2}+1}, \ &|k|<M^{j_{0}}, |n|<M^{j_{0}},\\
0, &\textup{otherwise}.
\end{cases}
\]
We have 
\begin{align}
&|\hat{u}_{j_{0}}(n,k)|<\varepsilon^{\frac{2}{3}-}e^{-\frac{3}{2}(|n|+|k|)^{c}},\notag\\
&|\partial\hat{u}_{j_{0}}(n,k)|<\varepsilon^{\frac{2}{3}-}e^{-\frac{3}{2}(|n|+|k|)^{c}}.\notag
\end{align}
By the definition of $\hat{u}_{j_{0}}$ and Plancherel's Theorem, we have
\begin{align}
&\lVert F(u_{j_{0}}) \rVert\notag\\
=&\lVert i\partial_{\omega}u_{j_{0}}-\Delta u_{j_{0}}+u_{j_{0}}+\varepsilon^{\frac{2}{3}} D^{\alpha} (|u_{j_{0}}|^{2}u_{j_{0}})-\varepsilon^{\frac{2}{3}}P(x,\theta)\rVert\notag\\
=&\lVert \varepsilon^{\frac{2}{3}}D^{\alpha}(|u_{j_{0}}|^{2}u_{j_{0}})-\varepsilon^{\frac{2}{3}}(1-\Gamma_{M^{j_{0}}}) P \rVert\notag\\
\le & \lVert \varepsilon^{\frac{2}{3}} D^{\alpha}(|u_{j_{0}}|^{2}u_{j_{0}}) \rVert+\lVert \varepsilon^{\frac{2}{3}}(1-\Gamma_{M^{j_{0}}}) P \rVert\notag\\
\le & \varepsilon+\varepsilon^{\frac{2}{3}}\exp(-(2-)(M^{j_{0}})^{c})\notag\\
\le &\exp(-2(M^{j_{0}})^{c}).
\end{align}
Moreover, we have
\begin{align}
    &\lVert \partial F(u_{j_{0}}) \rVert\notag\\
    =&\lVert \partial \varepsilon^{\frac{2}{3}}D^{\alpha}(|u_{j_{0}}|^{2}u_{j_{0}})-\partial\varepsilon^{\frac{2}{3}}(1-\Gamma_{M^{j_{0}}}) P \rVert\notag\\
    =&\lVert \partial \varepsilon^{\frac{2}{3}}D^{\alpha}(|u_{j_{0}}|^{2}u_{j_{0}})\rVert\notag\\
    \le & \varepsilon<\exp(-2(M^{j_{0}})^{c}).
\end{align}
By Whitney extension theorem, we can extend $\hat{u}_{j_{0}}(n,k)$ to the whole parameter space $[1,2]$ and we have
\begin{equation}
    \lVert \hat{u}_{j_{0}}(n,k) \rVert_{C^{1}([1,2])}<\varepsilon^{\frac{2}{3}-}e^{-(|n|+|k|)^{c}}.
\end{equation}
This completes the proof.
\end{proof}

\section{The Newton iteration scheme}
Assume that $u_{j}(x,\theta;\omega)\in L^{2}(\mathbb{T}^{d+1})$ $ (j\ge j_{0})$ is defined smoothly for $\omega\in[1,2]$. Moreover, assume that
\begin{itemize}
    \item[(j-i)] $\textup{supp } \hat{u}_{j}\subset B(0,M^{j})$.
    \item[(j-ii)] $\lVert u_{j} \rVert<(1+\sum\limits_{j'=j_{0}}^{j}\frac{1}{j'^{2}})\varepsilon^{1/5}$.
    \item[(j-iii)] $| \hat{u}_{j}(\xi)|< (1+\sum\limits_{j'=j_{0}}^{j}\frac{1}{j'^{2}}) e^{-|\xi|^{c}}$ and $|\partial \hat{u}_{j}(\xi)|< 2(1+\sum\limits_{j'=j_{0}}^{j}\frac{1}{j'^{2}}) e^{-|\xi|^{c}}$, where  $\xi=(n,k)\in\mathbb{Z}^{d+1}$, and $\partial$ refers to the derivative with respect to $\omega$.
    \item[(j-iv)] $\lVert F(u_{j})\rVert< e^{-2 (M^{j})^{c}}$ and $\lVert \partial F(u_{j})\rVert<  e^{-2 (M^{j})^{c}}$ are valid for $\omega$ restricted to a subset $\Lambda_{j}$ of $\omega$-parameter set which, in particular, we assume to be a union of disjoint intervals of size $\exp (-j^{C_{3}})$. 
\end{itemize}
 Let $u=u_{j}+v$, then we have
 \begin{align}
     F(u)
    &=i\partial_{\omega}u_{j}+i\partial_{\omega}v-\Delta u_{j}-\Delta v +u_{j}+v+\varepsilon^{\frac{2}{3}}D^{\alpha}\left( (u_{j}+v)^{2}(\bar{u}_{j}+\bar{v}) \right)-\varepsilon^{\frac{2}{3}} P\notag\\
    &=F(u_{j})+i\partial_{\omega}v-\Delta v+v+ \varepsilon^{\frac{2}{3}} D^{\alpha}(2 u_{j}\bar{u}_{j} v+u_{j}^{2}\bar{v})+\varepsilon^{\frac{2}{3}}D^{\alpha}(2u_{j}v\bar{v}+\bar{u}_{j}v^{2}+v^{2}\bar{v}).\label{zhenzhengdedongxi}
 \end{align}
Thus, we have the linearized equation 
\begin{equation}
    i\partial_{\omega}v-\Delta v+v+\varepsilon^{\frac{2}{3}}D^{\alpha}(2 u_{j}\bar{u}_{j} v+u_{j}^{2}\bar{v})=-F(u_{j}).
    \label{linearized equation}
\end{equation}
The equation conjugate to (\ref{linearized equation}) is
\begin{equation}
    -i \partial_{\omega}\bar{v}-\Delta \bar{v}+\bar{v}+\varepsilon^{\frac{2}{3}}D^{\alpha}(\bar{u}_{j}^{2}v+2u_{j}\bar{u}_{j} \bar{v})=-\overline{F}(u_{j}). 
\end{equation}
Let 
$$
w=\begin{pmatrix}
    v\\
    \bar{v}
\end{pmatrix},\ \ 
Q_{j}=\begin{pmatrix}
    -F(u_{j})\\
    -\overline{F}(u_{j})
\end{pmatrix}.
$$
Then, the homological equation is 
\begin{equation}
    \begin{pmatrix}
        i & 0\\
        0 & -i
    \end{pmatrix}\partial_{\omega} w-\Delta w+w+\varepsilon^{\frac{2}{3}} D^{\alpha} \begin{pmatrix}
        2u_{j}\bar{u}_{j} & u_{j}^{2}\\
        \bar{u}_{j}^{2} & 2u_{j}\bar{u}_{j}
    \end{pmatrix}w=Q_{j}.
\end{equation}
Passing to Fourier coefficients, we get a lattice system of equations:
\begin{equation}
    (D+\varepsilon^{\frac{2}{3}}\Lambda S_{j})\widehat{w}=\widehat{Q}_{j},
\end{equation}
where  
\begin{equation}
    D=\begin{pmatrix}
        -k\omega+|n|^{2}+1 & 0\\
        0 & k\omega+|n|^{2}+1
    \end{pmatrix},
\end{equation}
\begin{equation}
    \Lambda=\begin{pmatrix}
        k\cdot 0+\langle n\rangle^{\alpha} & 0\\
        0 & k\cdot 0+\langle n\rangle^{\alpha}
    \end{pmatrix},
\end{equation}
and
\begin{equation}
    S_{j}=\begin{pmatrix}
        S_{2 u_{j}\bar{u}_{j}} & S_{u_{j}^{2}}\\
        S_{ \bar{u}_{j}^{2}} & S_{2u_{j}\bar{u}_{j}}
    \end{pmatrix},
\end{equation}
where $S_{\phi}$ represents the Toeplitz operator corresponding to $\phi$ (i.e., $S_{\phi}((n,k),(n,k))=\hat{\phi}(n-n',k-k')$).  The index of the matrix $T_{j}=D+\varepsilon^{\frac{2}{3}}\Lambda S_{j}$ is $(\pm, n, k)\in \mathbb{Z}_{2}\times\mathbb{Z}^{d+1}$. Denote $T_{j,N}=T_{j}|_{|n|,|k|<N}$.

\begin{lemma}
    Let $N=M^{j+1}$. Under assumptions \textup{(j-i), (j-ii), (j-iii), (j-iv)}. Assume that we have confirmed  
    \begin{equation}
    \lVert T_{j,N}^{-1} \rVert<2\exp(\log N)^{C_{2}}:=B(N),
    \label{l2bound}
    \end{equation}
    and
    \begin{equation}
    |T_{j,N}^{-1}(\xi,\xi')|<2e^{-\frac{1}{2}|\xi-\xi'|^{c}} \textup{ for } |\xi-\xi'|>N^{\frac{1}{2}},
    \label{offdiagonalestimate}
    \end{equation}
    for $\omega$ restricted to a subset $\Lambda_{j}'\subset\Lambda_{j}$, which we assume to be a union of disjoint intervals of size $\frac{1}{1-e^{-j-10}} \exp(-(j+1)^{C_{3}})$. Then there exist a subset $\Lambda_{j+1}\subset\Lambda_{j}'$ and a function $u_{j+1}(x,\theta;\omega)\in L^{2}(\mathbb{T}^{d+1})$ defined smoothly for $\omega\in [1,2]$ such that
    \begin{itemize}
            \item[\textup{((j+1)-i)}] $\textup{supp } \hat{u}_{j+1}\subset B(0,M^{j+1})$.
            \item[\textup{((j+1)-ii)}] $\lVert u_{j+1} \rVert<(1+\sum\limits_{j'=j_{0}}^{j+1}\frac{1}{j'^{2}})\varepsilon^{1/5}$.
            \item[\textup{((j+1)-iii)}] $| \hat{u}_{j+1}(\xi)|< (1+\sum\limits_{j'=j_{0}}^{j+1}\frac{1}{j'^{2}}) e^{-|\xi|^{c}}$ and $|\partial \hat{u}_{j+1}(\xi)|< 2(1+\sum\limits_{j'=j_{0}}^{j+1}\frac{1}{j'^{2}}) e^{-|\xi|^{c}}$, where  $\xi=(n,k)\in\mathbb{Z}^{d+1}$, and $\partial$ refers to the derivative with respect to $\omega$.
            \item[\textup{((j+1)-iv)}] $\lVert F(u_{j+1})\rVert< e^{-2 (M^{j+1})^{c}}$ and $\lVert \partial F(u_{j+1})\rVert<  e^{-2 (M^{j+1})^{c}}$ are valid for $\omega$ restricted to  $\Lambda_{j+1}$ which is a union of disjoint intervals of size $\exp (-(j+1)^{C_{3}})$.
            \item[\textup{((j+1)-v)}] $\lVert \partial^{\beta}(u_{j+1}-u_{j}) \rVert<e^{-\frac{11}{8}(M^{j+1})^{c}}$, where $\beta=0,1$.
            \item[\textup{((j+1)-vi)}] $\textup{mes } (\Lambda_{j}'\setminus \Lambda_{j+1})\lesssim e^{-j-10}$.  
    \end{itemize}
    \label{lemma3.1}
\end{lemma}

\begin{proof}
For simplicity, we omit the subscript $j$ of $S_{j}$ and $T_{j}$.
    By the definition of $S$ and assumptions, we have 
\begin{equation}
    |S((\pm,n,k), (\pm, n', k'))| <  e^{-(1-)(|n-n'|+|k-k'|)^{c}},
\end{equation}
and
\begin{equation}
    |\partial S((\pm,n,k), (\pm, n', k'))|<  e^{-(1-)(|n-n'|+|k-k'|)^{c}}.
\end{equation}
 We omit the index $\pm$.
Thus, we have
\begin{equation}
    |T(\xi,\xi')|<\varepsilon^{2/3} \langle n\rangle^{\alpha} e^{-(1-)|\xi-\xi'|^{c}}, \textup{ for } \xi\ne\xi', \xi=(n,k),
    \label{shuaijianyuan}
\end{equation}
and
\begin{equation}
    |\partial T(\xi,\xi')|<\varepsilon^{2/3} \langle n\rangle^{\alpha} e^{-(1-)|\xi-\xi'|^{c}},  \textup{ for } \xi\ne\xi', \xi=(n,k).
    \label{shuajiandaoshu}
\end{equation}
Let $\omega\in\Lambda_{j}'$.
Since one has
\begin{equation}
    \partial T_{N}^{-1}=-T_{N}^{-1} (\partial T_{N}) T_{N}^{-1},
\end{equation}
it follows that
\begin{align}
    \lVert \partial T_{N}^{-1} \rVert
    &\le \lVert T_{N}^{-1} \rVert^{2} \lVert \partial T_{N} \rVert\notag\\
    &<B(N)^{2} N <e^{3(\log N)^{C_{2}}}.
\end{align}
When $|\xi-\xi'|>N^{3/4}$, we have
\begin{align}
    |(\partial T_{N}^{-1})(\xi,\xi')|&\le \sum\limits_{|\xi_{1}|,|\xi_{2}|<N}\left| T_{N}^{-1}(\xi,\xi_{1})  \right| \left| \partial T_{N}(\xi_{1},\xi_{2}) \right| \left| T_{N}^{-1}(\xi_{2}, \xi') \right|\notag\\
        &\le N^{2(d+1)+3} e^{5(\log N)^{C_{2}}} e^{-\frac{1}{2}(|\xi-\xi'|-3 N^{1/2})^{c}}\notag\\
        &\le e^{-(\frac{1}{2}-)|\xi-\xi'|^{c}}.\label{offdiagonaldecayfordaoshu}
\end{align}
Since $\hat{Q}_{j}$ is expressed as a polynomial in $\hat{u}_{j}$ and $\hat{\bar{u}}_{j}$, and given that $\textup{supp }\hat{u}_{j}\subset B(0, M^{j})$, we get
\begin{equation}
    \textup{supp } \hat{Q}_{j}\subset B(0, C M^{j})\subset B(0,\frac{1}{4}M^{j+1}).
\end{equation}
Let 
\begin{equation}
    \hat{w}_{j}=\begin{pmatrix}
    \hat{v}_{j}\\
    \hat{\bar{v}}_{j}
\end{pmatrix}=T_{N}^{-1} \hat{Q}_{j},
\label{solutionineachstep}
\end{equation}
and 
$$
u_{j+1}=u_{j}+v_{j}.
$$
It is straightforward that 
\begin{equation}
    \textup{supp } \hat{v}_{j}\subset B(0, M^{j+1}).
    \label{zhiji}
\end{equation}
By (\ref{solutionineachstep}), we have
\begin{align}
    \lVert v_{j} \rVert&\le \lVert T_{N}^{-1} \rVert \lVert Q_{j} \rVert \notag\\
        &\le \exp(\log N)^{C_{2}} e^{-2(M^{j})^{c}}\notag\\
        &<e^{-(2-) (M^{j})^{c}}<e^{-\frac{3}{2} (M^{j+1})^{c}},\label{feidaoshuguji}
\end{align}
provided that $c<\frac{\log \frac{4}{3}}{\log M}$.
Moreover, we have
\begin{align}
    \lVert \partial v_{j} \rVert &\le \lVert \partial T_{N}^{-1} \rVert \lVert Q_{j} \rVert+\lVert T_{N}^{-1} \rVert \lVert \partial Q_{j} \rVert\notag\\
    &\le e^{3(\log N)^{C_{2}}}e^{-2 (M^{j})^{c}}+e^{(\log N)^{C_{2}}} e^{-2(M^{j})^{c}}\notag\\
    &<e^{-(2-) (M^{j})^{c}}<e^{-\frac{3}{2} (M^{j+1})^{c}}.\label{daoshuguji}
\end{align}
Thus, (\ref{zhiji}), (\ref{feidaoshuguji}) and (\ref{daoshuguji}) permit us to confirm ((j+1)-i), ((j+1)-ii), ((j+1)-iii), ((j+1)-v) for $u_{j+1}=u_{j}+v_{j}$ provided the $\omega$-parameters are restricted to $\Lambda_{j}'$. 

By (\ref{zhenzhengdedongxi}), we have
\begin{equation}
    \begin{split}
        \begin{pmatrix}
            F(u_{j+1})\\
            \overline{F}(u_{j+1})
        \end{pmatrix}
        =\left((T-T_{N})
        \hat{w}_{j} \right)^{\vee}
        +\varepsilon^{2/3}\begin{pmatrix}
            D^{\alpha} (2 u_{j} v_{j} \bar{v}_{j}+\bar{u}_{j}v_{j}^{2}+v_{j}^{2}\bar{v}_{j})\\
            D^{\alpha} (2 \bar{u}_{j} v_{j} \bar{v}_{j}+u_{j}\bar{v}_{j}^{2}+v_{j}\bar{v}_{j}^{2})
        \end{pmatrix},
    \end{split}
\end{equation}
where $\hat{w}_{j}=\begin{pmatrix}
        \hat{v}_{j}\\
        \hat{\bar{v}}_{j}
        \end{pmatrix}$.
For the second term, we have
\begin{equation}
    \lVert D^{\alpha} (2 u_{j} v_{j} \bar{v}_{j}+\bar{u}_{j}v_{j}^{2}+v_{j}^{2}\bar{v}_{j}) \rVert < e^{-(3-)(M^{j+1})^{c}}
\end{equation}
and 
\begin{equation}
    \begin{split}
        \lVert \partial D^{\alpha} (2 u_{j} v_{j} \bar{v}_{j}+\bar{u}_{j}v_{j}^{2}+v_{j}^{2}\bar{v}_{j}) \rVert
        = \lVert D^{\alpha} \partial (2 u_{j} v_{j} \bar{v}_{j}+\bar{u}_{j}v_{j}^{2}+v_{j}^{2}\bar{v}_{j}) \rVert < e^{-(3-)(M^{j+1})^{c}}.
    \end{split}
\end{equation}

Now we consider the first term. Denote $P_{K}$ the projection on $B(0,K)$. We have
\begin{align}
    (T-T_{N})\hat{w}_{j}
    &=(I-P_{N})T P_{\frac{N}{2}}\hat{w}_{j} +(T-T_{N})(\hat{w}_{j}-P_{\frac{N}{2}}\hat{w}_{j})\notag\\
    &=(I-P_{N})T P_{\frac{N}{2}}\hat{w}_{j}+(T-T_{N})(I-P_{\frac{N}{2}})T_{N}^{-1} \hat{Q}_{j}\notag\\
    &=(I-P_{N})T P_{\frac{N}{2}}\hat{w}_{j}+(T-T_{N})(I-P_{\frac{N}{2}})T_{N}^{-1} P_{\frac{N}{4}}\hat{Q}_{j}.
\end{align}
By (\ref{shuaijianyuan}), (\ref{shuajiandaoshu}), (\ref{feidaoshuguji}) and (\ref{daoshuguji}), we obtain
\begin{align}
    \lVert \partial^{\beta} (I-P_{N})T P_{\frac{N}{2}}\hat{w}_{j} \rVert
    &< N^{\alpha} e^{-(1-)(\frac{1}{2} N)^{c}} e^{-\frac{3}{2} (M^{j+1})^{c}}\notag\\
    &<\frac{1}{3} e^{-\frac{1}{2} N^{c}}e^{-\frac{3}{2} (M^{j+1})^{c}}=\frac{1}{3} e^{-2 (M^{j+1})^{c}},
\end{align}
for $\beta=0,1$.
By (j-iv), (\ref{offdiagonalestimate}) and (\ref{offdiagonaldecayfordaoshu}), we get
\begin{align}
     \lVert \partial^{\beta} (T-T_{N})(I-P_{\frac{N}{2}})T_{N}^{-1} P_{\frac{N}{4}}\hat{Q}_{j} \rVert
        &< e^{-(\frac{1}{2}-)(\frac{1}{4} N)^{c}} e^{-2(M^{j})^{c}}\notag\\
        &\le \frac{1}{3}e^{-\frac{17}{8} (M^{j})^{c}}\notag\\
        &\le \frac{1}{3}e^{-2 (M^{j+1})^{c}},
\end{align}
for $\beta=0,1$, provided that $c<\frac{\log \frac{17}{16}}{\log M}$.
Hence, we obtain
\begin{equation}
    \lVert \partial^{\beta} (T-T_{N})\hat{w}_{j} \rVert<\frac{2}{3}e^{-2 (M^{j+1})^{c}}, 
    \label{j+1buiv}
\end{equation}
for $\beta=0,1$ and $\omega\in\Lambda_{j}'$. Thus, we have
\begin{equation}
    \lVert  \partial^{\beta} F(u_{j+1})  \rVert<e^{-2(M^{j+1})^{c}},
\end{equation}
for $\beta=0,1$ and $\omega\in\Lambda_{j}'$.
Note that in the above $v_{j}$ is defined on $\Lambda_{j}'$. We need to extend its definition to the entire $\omega$-parameter set $[1,2]$. Assume 
\begin{equation}
    \Lambda_{j}' =\mathop{\cup}_{\zeta} I_{\zeta},
\end{equation}
where  $\{I_{\zeta}\}$ are disjoint intervals of size $\frac{1}{1-e^{-j-10}}\exp(-(j+1)^{C_{3}})$. For each $\zeta$, denote $I_{\zeta}'\subset I_{\zeta}$ the $(1-e^{-j-10})$-dilation of $I_{\zeta}$ with the same center. Let $0\le \psi_{\zeta}\le 1$ be a smooth function satisfying
\begin{align}
    &\psi_{\zeta}=0 \textup{ outside } I_{\zeta}, \ \psi_{\zeta}=1 \textup{ on } I_{\zeta}',\\
    &|\partial \psi_{\zeta}|<\exp(2(j+1)^{C_{3}}).
\end{align}
Let 
\begin{equation}
    \tilde{v}_{j}=\sum\limits_{\zeta} \psi_{\zeta} v_{j},
\end{equation}
then we have
\begin{align}
&\textup{supp } \hat{\tilde{v}}_{j}\subset B(0, M^{j+1}), \label{zhijiyantuo}\\
&\lVert \partial^{\beta} \tilde{v}_{j} \rVert < e^{-(\frac{3}{2}-)(M^{j+1})^{c}}, \textup{ for } \beta=0,1.\label{gujiyantuo}
\end{align}
Let $u_{j+1}=u_{j}+\tilde{v}_{j}$.  Thus, (\ref{zhijiyantuo}) and (\ref{gujiyantuo}) permit us to confirm ((j+1)-i), ((j+1)-ii), ((j+1)-iii), ((j+1)-v). Let $\Lambda_{j+1}=\mathop{\cup}\limits_{\zeta} I_{\zeta}'$, then ((j+1)-iv) is confirmed. Moreover, we have
\begin{equation}
\textup{mes } (\Lambda_{j}'\setminus \Lambda_{j+1})\lesssim e^{-j-10}.
\end{equation}
This completes the proof.
\end{proof}

Thus, the key is to obtain (\ref{l2bound}) and (\ref{offdiagonalestimate}). We refer to $T_{N}^{-1}$ as Green function.

\section{Estimate of Green function}
In this section, we estimate $T_{j,N}^{-1}$ and find $\Lambda_{j}'$ step by step.
Recall that
\begin{equation}
    T_{j}=D+\varepsilon^{2/3}\Lambda S_{j}.
\end{equation}
Denote
\begin{equation}
    \tilde{T}_{j}=\Lambda^{-1} D+\varepsilon^{2/3} S_{j}.
\end{equation}
We have $T_{j}=\Lambda \tilde{T}_{j}$. Thus, we have $T_{j,N}^{-1}=\tilde{T}_{j,N}^{-1} \Lambda_{N}^{-1}$.
To obtain (\ref{l2bound}) and  (\ref{offdiagonalestimate}), it suffices to ensure that
\begin{align}
    &\lVert \tilde{T}_{j,N}^{-1} \rVert <2\exp(\log N)^{C_{2}} \label{l2bound2},\\
    &|\tilde{T}_{j,N}^{-1}(\xi,\xi')|<2e^{-\frac{1}{2}|\xi-\xi'|^{c}} \textup{ for } |\xi-\xi'|>N^{\frac{1}{2}} \label{offdiagonal2}.
\end{align}
The reason we introduce $\tilde{T}$ is that it is self-adjoint.
Denote
$\tilde{D}=\Lambda^{-1} D$,
$\tilde{T}_{j}=\tilde{D}+\varepsilon^{2/3} S_{j}$ and $N_{j}=M^{j}$.  
Furthermore, we denote
\begin{equation}
    \tilde{D}_{\pm,n,k}=\langle  n\rangle^{-\alpha}(\pm k\omega+|n|^{2}+1).
\end{equation}

For $N_{j+1}\le \varepsilon^{-\frac{1}{30d}}$,  we construct $\Lambda_{j}$ and $\Lambda_{j}'$ ($j\ge j_{0}$) directly.
Note that $\Lambda_{j_{0}}$ has been constructed in Section 2. Now we construct $\Lambda_{j_{0}}'$. Partition $[1,2]$ into intervals $\{I_{j_{0+1},\zeta}\}$ of size $\frac{1}{1-e^{-j_{0}-10}}\exp (-(j_{0}+1)^{C_{3}})$. Let 
\begin{align}
         &\Lambda_{j_{0}}'\notag\\
         =&\bigcup \{ I_{j_{0}+1,\zeta}: \exists\  \omega'\in I_{j_{0}+1,\zeta} \textup{ such that } 
         |-k\omega'+|n|^{2}+1|\ge (\log \frac{1}{\varepsilon})^{-1}(1+|k|)^{-\tau},\notag\\
         &\ \ \  \ \ \      |n|<N_{j_{0}+1}, |k|<N_{j_{0}+1}\}.\notag
    \end{align}
Note that $\frac{2}{1-e^{-j_{0}-10}}\exp (-(j_{0}+1)^{C_{3}})<\exp (-j_{0}^{C_{3}})$. It is easy to check that $\Lambda_{j_{0}}'\subset \Lambda_{j_{0}}$.

For $\omega\in\Lambda_{j_{0}}'$, we have 
\begin{align}
    |-k\omega+|n|^{2}+1|&\ge (\log \frac{1}{\varepsilon})^{-1}(1+|k|)^{-\tau}-2N_{j_{0}+1} \cdot \frac{1}{1-e^{-j_{0}-10}}\exp (-(j_{0}+1)^{C_{3}})\notag\\
    &>\frac{1}{2} (\log \frac{1}{\varepsilon})^{-1}(1+|k|)^{-\tau},\notag
\end{align}
for $|n|<N_{j_{0}+1}, |k|<N_{j_{0}+1}$.
Thus, we have
\begin{equation}
    |\tilde{D}_{\pm,n,k}|>\frac{1}{2}(\log \frac{1}{\varepsilon})^{-1} N_{j_{0}+1}^{-\tau-\alpha},
\end{equation}
for $|n|<N_{j_{0}+1}, |k|<N_{j_{0}+1}$.
We obtain
\begin{equation}
\begin{split}
     \lVert \tilde{T}_{j_{0}, N_{j_{0}+1}}^{-1} \rVert
     &=\lVert  (I+\varepsilon^{2/3} \tilde{D}_{N_{j_{0}+1}}^{-1} S_{j_{0}, N_{j_{0}+1}})^{-1} \tilde{D}_{N_{j_{0}+1}}^{-1} \rVert\\
     &\le 4 (\log \frac{1}{\varepsilon}) N_{j_{0}+1}^{\tau+\alpha} <\exp(\log N_{j_{0}+1})^{C_{2}}.
\end{split}
\end{equation}
Furthermore, we have
\begin{equation}
    \tilde{T}_{j_{0}, N_{j_{0}+1}}^{-1}=\tilde{D}_{N_{j_{0}+1}}^{-1}+\sum\limits_{l=1}^{\infty} (-1)^{l} \varepsilon^{\frac{2l}{3}} (\tilde{D}_{N_{j_{0}+1}}^{-1} S_{j_{0}, N_{j_{0}+1}})^{l} \tilde{D}_{N_{j_{0}+1}}^{-1}.
\end{equation}
Note that we have
\begin{align}
    \left|\left((\tilde{D}_{N_{j_{0}+1}}^{-1} S_{j_{0}, N_{j_{0}+1}})^{l} \tilde{D}_{N_{j_{0}+1}}^{-1}\right)(\xi,\xi')\right|
    &\le \sum\limits_{|\xi_{1}|,...,|\xi_{l-1}|\le N_{j_{0}+1}} N_{j_{0}+1}^{10d l} e^{-(1-)|\xi-\xi'|^{c}}\notag\\
    &\le N_{j_{0}+1}^{12dl} e^{-(1-)|\xi-\xi'|^{c}}.
\end{align}
Thus, we obtain
\begin{align}
    |\tilde{T}_{j_{0}, N_{j_{0}+1}}^{-1}(\xi,\xi')|
    &\le \sum\limits_{l=1}^{\infty} (-\varepsilon^{\frac{2}{3}} N_{j_{0}+1}^{12d})^{l} e^{-(1-)|\xi-\xi'|^{c}}\notag\\
    &\le e^{-\frac{1}{2}|\xi-\xi'|^{c}}, \textup{ for } \xi\ne\xi'.
\end{align}
By Lemma \ref{lemma3.1}, we obtain $\Lambda_{j_{0}+1}$.
Proceed this process until $N_{j+1}> \varepsilon^{-\frac{1}{30d}}$. We obtain
\begin{equation}
    \Lambda_{j}\subset \Lambda_{j-1}'\subset \Lambda_{j-1}\subset \cdots \subset \Lambda_{j_{0}}'\subset\Lambda_{j_{0}}.\notag
\end{equation}
Lemma \ref{lemma3.1} and a standard measure estimate as in Section 2 tell us 
\begin{equation}
    \textup{mes } [1,2]\setminus \Lambda_{j}\lesssim (\log \frac{1}{\varepsilon})^{-1}+e^{-j_{0}-10}.
\end{equation}

For $N_{j+1}> \varepsilon^{-\frac{1}{30d}}$, we have the following lemma. 

\begin{lemma}
    Let $N_{j+1}> \varepsilon^{-\frac{1}{30d}}$. Assume there exist $\Lambda_{j}\subset \Lambda_{j-1}\subset \cdots \subset \Lambda_{j_{0}}\subset [1,2]$ and functions $u_{j'}(x,\theta;\omega)\in L^{2}(\mathbb{T}^{d+1}) (j_{0}\le j'\le j)$ defined smoothly for $\omega\in [1,2]$ such that ($j'$-i), ($j'$-ii), ($j'$-iii), ($j'$-iv) $(j_{0}\le j'\le j)$ and ($j'$-v) $(j_{0}+1\le j'\le j)$ hold.
    Moreover, assume we have
    \begin{align}
    &\lVert \tilde{T}_{j', N_{j'+1}}^{-1} \rVert <2\exp(\log N_{j'+1})^{C_{2}} \label{l2bound2j'},\\
    &|\tilde{T}_{j', N_{j'+1}}^{-1}(\xi,\xi')|<2e^{-\frac{1}{2}|\xi-\xi'|^{c}} \textup{ for } |\xi-\xi'|>N_{j'+1}^{\frac{1}{2}} \label{offdiagonal2j'}.
\end{align}
for $\omega\in\Lambda_{j'+1}$ for $j_{0}\le j'\le j-1$. Then, there exists a subset $\Lambda_{j}'\subset \Lambda_{j}$ which is a union of disjoint intervals of size $\frac{1}{1-e^{-j-10}}\exp (-(j+1)^{C_{3}})$ such that 
\begin{align}
    &\lVert \tilde{T}_{j, N_{j+1}}^{-1} \rVert<2\exp (\log N_{j+1})^{C_{2}},\notag\\
    &|\tilde{T}_{j, N_{j+1}}^{-1}(\xi,\xi')|<2e^{-\frac{1}{2}|\xi-\xi'|^{c}}, \textup{ for } |\xi-\xi'|>N_{j+1}^{\frac{1}{2}}.
\end{align}
Moreover, we have
\begin{equation}
    \textup{mes } \Lambda_{j}\setminus\Lambda_{j}'\le N_{j+1}^{-10d}.
\end{equation}
\label{lemma4.1}
\end{lemma}
\begin{proof}
    To control $\tilde{T}_{j, N_{j+1}}^{-1}$, we cover $Q=[-N_{j+1}, N_{j+1}]^{d+1}$ by intervals $Q_{0}=[-N_{\lfloor\frac{j+1}{3}\rfloor}, N_{\lfloor\frac{j+1}{3}\rfloor}]$ and intervals  $Q_{r}$ in $\mathbb{Z}^{d+1}$ of size $N_{j+1}^{1/3}$ such that $\textup{dist } (0, Q_{r})>N_{j+1}^{1/4}$.
By the inductive hypothesis, we have
\begin{align}
    &\lVert \tilde{T}_{\lfloor (j+1)/3\rfloor-1, N_{\lfloor (j+1)/3\rfloor}}^{-1} \rVert <2\exp(\log N_{\lfloor (j+1)/3\rfloor})^{C_{2}} \label{l2bound2lfloorj2rfloor},\\
    &|\tilde{T}_{\lfloor (j+1)/3\rfloor-1, N_{\lfloor (j+1)/3\rfloor}}^{-1}(\xi,\xi')|<2 e^{-\frac{1}{2}|\xi-\xi'|^{c}} \textup{ for } |\xi-\xi'|>N_{\lfloor (j+1)/3\rfloor}^{\frac{1}{2}} \label{offdiagonal2lfloorj2rfloor}.
\end{align}
Since $\lVert v_{j'} \rVert<e^{-\frac{11}{8}(M^{j'+1})^{c}}$, we have
\begin{equation}
\lVert (u_{j}-u_{\lfloor (j+1)/3 \rfloor-1})(\xi) \rVert<e^{-(\frac{3}{8}-)(M^{\lfloor (j+1)/3 \rfloor})^{c}} e^{-|\xi|^{c}}.    
\end{equation}
Thus, we have
\begin{equation}
    |(\tilde{T}_{j}-\tilde{T}_{\lfloor (j+1)/3 \rfloor-1})(\xi,\xi')|<e^{-\frac{1}{4}(M^{\lfloor (j+1)/3 \rfloor})^{c}} e^{-|\xi-\xi'|^{c}}.
    \label{chajv}
\end{equation}
Form (\ref{l2bound2lfloorj2rfloor}), (\ref{offdiagonal2lfloorj2rfloor}), (\ref{chajv}) and $e^{-\frac{1}{4}(M^{\lfloor (j+1)/3 \rfloor})^{c}} N_{\lfloor (j+1)/3 \rfloor}^{C} \exp(\log N_{\lfloor (j+1)/3 \rfloor})^{C_{2}}\ll 1$, we obtain
\begin{align}
    &\lVert \tilde{T}_{j, N_{\lfloor (j+1)/3\rfloor}}^{-1} \rVert <4\exp(\log N_{\lfloor (j+1)/3\rfloor})^{C_{2}} \label{l2bound2lfloorj2rfloorj},\\
    &|\tilde{T}_{j, N_{\lfloor (j+1)/3\rfloor}}^{-1}(\xi,\xi')|<4e^{-\frac{1}{2}|\xi-\xi'|^{c}} \textup{ for } |\xi-\xi'|>N_{\lfloor (j+1)/3\rfloor}^{\frac{1}{2}} \label{offdiagonal2lfloorj2rfloorj}.
\end{align}

To estimate $\tilde{T}_{j, Q_{r}}^{-1}$, we need the following  lemma:
\begin{lemma}
    Fix any large number $B$. There is a partition $\{\pi_{\zeta}\}$ of $\mathbb{Z}^{d}$ satisfying the properties 
    $$
    \textup{diam } \pi_{\zeta}< B^{\tilde{C}_{d}}
    $$
    and 
    $$
    |n-n'|+||n|^{2}-|n'|^{2}|>B, \textup{ if } n\in \pi_{\zeta}, n'\in\pi_{\zeta'}, \zeta\ne\zeta'.
    $$
    Here, $\tilde{C}_{d}=(2d+2)^{d+2}$. 
    \label{fenlidingli}
\end{lemma}

In the sequel, we omit the subscript $j$ of $T_{j}$ and the subscript $j+1$ of $N_{j+1}$, i.e., we denote $T_{j}=T$ and $N_{j+1}=N$. Fix $\omega$.
Denote 
\begin{equation}
    \Omega_{r}=\{(n,k)\in Q_{r}: \min\limits_{\pm}|\tilde{D}_{\pm,n,k}|=\min\limits_{\pm}|\langle n \rangle^{-\alpha} (\pm k\omega+|n|^{2}+1)|<1\}.
\end{equation}
Choose a constant $\delta$ such that $\delta>\alpha$ and $(d\tilde{C}_{d}+2)\delta<\frac{1}{30}$, which can be chosen since we assume $\alpha<\frac{1}{30(d(2d+2)^{d+2}+2)}$. Let $B=N^{\delta}$ in Lemma \ref{fenlidingli}, and denote by $\{\pi_{\zeta}\}$ the partition of $\mathbb{Z}^{d}$. 
By Lemma \ref{fenlidingli}, we have
\begin{equation}
    \textup{diam } \pi_{\zeta}< N^{\delta \tilde{C}_{d}},
    \label{sizedaxiao}
\end{equation}
and
\begin{equation}
    |n-n'|+||n|^{2}-|n'|^{2}|>N^{\delta}, \textup{ if } n\in \pi_{\zeta}, n'\in\pi_{\zeta'}, \zeta\ne\zeta'.
    \label{jianjv}
\end{equation}
Now we consider the structure of $\Omega_{r}$.
If $(n,k), (n',k')\in \Omega_{r}$, we have
\begin{equation}
    ||n|^{2}-|k\omega||< N^{\alpha}+1 \textup{ and } ||n'|^{2}-|k'\omega||<N^{\alpha}+1.
\end{equation}
Thus, we have
\begin{equation}
    ||n|^{2}-|n'|^{2}|<2|k-k'|+2 N^{\alpha}+2.
\end{equation}
If $n\in \pi_{\zeta}$, $n'\in \pi_{\zeta'}$, $\zeta\ne \zeta'$, then
\begin{equation}
    2|k-k'|+|n-n'|+2N^{\alpha}+2>|n-n'|+||n|^{2}-|n'|^{2}|>N^{\delta},
\end{equation}
which implies
\begin{equation}
    |k-k'|+|n-n'|\gtrsim N^{\delta}.
\end{equation}
Fix $|n|<N$, the number  of $k$ such that $\min\limits_{\pm}|\langle n \rangle^{-\alpha} (\pm k\omega+|n|^{2}+1)|<1$ is at most $4N^{\alpha}$. Hence, we have
\begin{equation}
    \#\{(n,k)\in \Omega_{r}: n\in\pi_{\zeta}\}\lesssim N^{\delta \tilde{C}_{d} d}\cdot N^{\alpha}.
\end{equation}
The above argument allows us to obtain a partition of $\Omega_{r}$:
\begin{equation}
    \Omega_{r}=\mathop{\cup}\limits_{\kappa} \Omega_{r,\kappa},
\end{equation}
such that
\begin{align}
    &\textup{diam } \Omega_{r, \kappa}\lesssim N^{\delta \tilde{C}_{d} d+\alpha+\delta}<N^{(\tilde{C}_{d}d+2)\delta},\\
    &\textup{dist }(\Omega_{r,\kappa},\Omega_{r,\kappa'})\gtrsim N^{\delta}.
\end{align}
Let $\omega'$ in the $O(N^{-2})$-neighborhood of $\omega$.  We have 
\begin{equation}
    \min\limits_{\pm}|\langle n \rangle^{-\alpha} (\pm k\omega'+|n|^{2}+1)|>\frac{1}{2}, \textup{ for } (n,k)\notin \Omega_{r}.
\end{equation}
Hence, $\{\Omega_{r,\kappa}\}$ remains fixed for $\omega'$ in the $O(N^{-2})$-neighborhood of $\omega$. Denote this neighborhood by $I_{r,\kappa, s}$ ($s<N^{2}$, and the subscript $s$ corresponds to the neighborhood of $\omega$).  Let $\tilde{\Omega}_{r,\kappa,s}$ be the $N^{\frac{\delta}{2}}$-neighborhood of $\Omega_{r,\kappa,s}$.
Note that $(d\tilde{C}_{d}+2)\delta<\frac{1}{30}$.
As long as we  ensure that 
\begin{equation}
    \lVert \tilde{T}_{\tilde{\Omega}_{r,\kappa,s }}^{-1} \rVert<N^{C_{1}}
\end{equation}
for all $\kappa$,  applying Lemma \ref{coupling lemma2} gives us
\begin{align}
    &\lVert \tilde{T}_{Q_{r}}^{-1} \rVert\lesssim N^{C_{1}+1}, \label{l2boundQr}  \\
    &| \tilde{T}_{Q_{r}}^{-1}(\xi,\xi')|<e^{-\frac{1}{10}|\xi-\xi'|^{c}}, \textup{for } |\xi-\xi'|>N^{\frac{1}{5}}. \label{offdiagonalQr}
\end{align}
Let $\omega_{1}'=\omega'^{-1}$ and write
\begin{equation}
    \omega_{1}' \tilde{T}_{\tilde{\Omega}_{r,\kappa, s}}=
    \begin{pmatrix}
        \langle n \rangle^{-\alpha}(-k+\omega_{1}' (|n|^{2}+1)) & 0\\
        0 & \langle n \rangle^{-\alpha}(k+\omega_{1}'(|n|^{2}+1)) 
    \end{pmatrix}+\varepsilon^{2/3} \omega_{1}' S_{\tilde{\Omega}_{r,\kappa}}.
\end{equation}
Then, we have
\begin{equation}
    \partial_{\omega_{1}'} (\omega_{1}' \tilde{T}_{\tilde{\Omega}_{r,\kappa, s}})=
    \begin{pmatrix}
        (|n|^{2}+1)\langle n \rangle^{-\alpha}&0\\
        0 & (|n|^{2}+1)\langle n \rangle^{-\alpha}
    \end{pmatrix}+O(\varepsilon^{2/3}).
\end{equation}
Denote $E_{\tilde{\Omega}_{r,\kappa, s}}(\omega_{1}')$ as an eigenvalue of $\omega_{1}'\tilde{T}_{\tilde{\Omega}_{r,\kappa, s}}$.
The first-order eigenvalue variation implies that
\begin{equation}
    |\partial_{\omega_{1}'}  E_{\tilde{\Omega}_{r,\kappa, s}}|\gtrsim 1.
\end{equation}
Thus, there exists a subset $\tilde{I}_{r,\kappa, s}$  of $I_{r,\kappa, s}$ satisfying
\begin{itemize}
    \item $\textup{mes } \tilde{I}_{r,\kappa, s} \lesssim N^{1/2} N^{-C_{1}}$.
    \item For $\omega'\in I_{r,\kappa, s} \setminus \tilde{I}_{r,\kappa, s}$, we have
    \begin{equation}
        \lVert \tilde{T}_{\tilde{\Omega}_{r,\kappa, s}}^{-1} \rVert<N^{C_{1}}.
        \notag
    \end{equation}
\end{itemize}
Thus, for $\omega\in \Lambda_{j-1}\setminus \mathop{\cup}\limits_{r,\kappa,s} \tilde{I}_{r,\kappa, s}$, we have
\begin{equation}
    \lVert \tilde{T}_{\tilde{\Omega}_{r,\kappa}}^{-1} \rVert<N^{C_{1}}, \textup{ for all } r,\kappa.
\end{equation}
By applying Lemma \ref{coupling lemma2}, we have (\ref{l2boundQr}) and (\ref{offdiagonalQr}) for all $r$. Furthermore, by applying Lemma \ref{coupling lemma1}, we have
\begin{align}
    &\lVert \tilde{T}_{j, N_{j+1}}^{-1} \rVert <\exp(\log N_{j+1})^{C_{2}} \label{l2bound2j},\\
    &|\tilde{T}_{j, N_{j+1}}^{-1}(\xi,\xi')|<e^{-\frac{1}{2}|\xi-\xi'|^{c}} \textup{ for } |\xi-\xi'|>N_{j+1}^{\frac{1}{2}} \label{offdiagonal2j},
\end{align}
for $\omega\in \Lambda_{j-1}\setminus \mathop{\cup}\limits_{r,\kappa,s} \tilde{I}_{r,\kappa, s}$.
Now we construct $\Lambda_{j}'$. Partition $\Lambda_{j}$ into intervals $\{I_{j,\zeta}\}$ of size $\frac{1}{1-e^{-j-10}}\exp(-(j+1)^{C_{3}})$.
Let 
\begin{equation}
    \Lambda_{j}'=\cup \{I_{j,\zeta}: \exists\ \omega\in I_{j,\zeta}  \textup{ such that } (\ref{l2bound2j}), (\ref{offdiagonal2j}) \}.
\end{equation}
By a perturbation argument, we have
\begin{align}
    &\lVert \tilde{T}_{j, N_{j+1}}^{-1} \rVert <2\exp(\log N_{j+1})^{C_{2}},\\
    &|\tilde{T}_{j, N_{j+1}}^{-1}(\xi,\xi')|<2 e^{-\frac{1}{2}|\xi-\xi'|^{c}} \textup{ for } |\xi-\xi'|>N_{j+1}^{\frac{1}{2}}, 
\end{align}
for $\omega\in\Lambda_{j}'$.
Moreover, we have
    \begin{align}
         \textup{mes } \Lambda_{j}\setminus \Lambda_{j}'&<\textup{mes } \mathop{\cup}\limits_{r,\kappa,s} \tilde{I}_{r,\kappa, s}\notag\\
        &<N_{j+1}^{2d} N_{j+1}^{2} N_{j+1}^{-C_{1}+\frac{1}{2}}<N_{j+1}^{-C_{1}+2d+3}<N_{j+1}^{-10d},  
    \end{align}
     where the second inequality follows from the fact that the number of both $r$ and $\kappa$ is less than $N_j^{d}$.
\end{proof}

 Theorem \ref{maintheorem} follows from Lemma \ref{lemma2.3}, Lemma \ref{lemma3.1} and Lemma \ref{lemma4.1}.

\appendix
\section{Coupling Lemma}
In this appendix, we state two coupling lemmas whose proof can be found in \cite[Lemma 5.3, Lemma 7]{bourgain1998quasi}.
\begin{lemma}
Assume $T$ satisfies 
\begin{equation}
    |T(\xi,\xi')|<e^{-|\xi-\xi'|^{c}} \textup{ for } \xi\ne\xi'.
\end{equation}
Let $\Omega\subset \mathbb{Z}^{d}$ be an interval and assume $\Omega=\mathop{\cup}\limits_{\zeta} \Omega_{\zeta}$ a covering of $\Omega$ with intervals $\Omega_{\zeta}$ satisfying
\begin{itemize}
    \item $|T_{\Omega_{\zeta}}^{-1}(\xi,\xi')|<B$.
    \item $|T_{\Omega_{\zeta}}^{-1}(\xi,\xi')|<K^{-C}$ for $|\xi-\xi'|>\frac{K}{100}$.
    \item For each $\xi\in\Omega$, there is $\zeta$ such that 
    \begin{equation}
        B_{K}(\xi)\cap \Omega=\{\xi'\in\Omega:|\xi'-\xi|\le K\}\subset \Omega_{\zeta}.
    \end{equation}
    \item $\textup{diam } \Omega_{\zeta}<C' K$ for each $\zeta$. 
\end{itemize}
Here, $C>C(d)$ and $B, K$ are numbers satisfying the relation
    \begin{equation}
        \log B<\frac{1}{100} K^{c} \textup{ and } K>K_{0}(c, C', r).
    \end{equation}
Then 
\begin{align}
    &|T_{\Omega}^{-1}(\xi,\xi')|<2B,\\
    &|T_{\Omega}^{-1}(\xi,\xi')|<e^{-\frac{1}{2}|x-y|^{c}} \textup{ for } |\xi-\xi'|>(100 C'K)^{\frac{1}{1-c}}.
\end{align}
    \label{coupling lemma1}
\end{lemma}

\begin{lemma}
    Fix some constants $\frac{1}{10}>\varepsilon_{1}>\varepsilon_{2}>\varepsilon_{3}>0$ and let $\Omega$ be a subset of the $M$-ball in $\mathbb{Z}^{d+1}$ ($M\to \infty$). Assume $\{\Omega_{\kappa}\}$ a collection of subsets of $\Omega$ satisfying 
    \begin{align}
        \textup{diam } \Omega_{\kappa}<M^{\varepsilon_{1}},\\
        \textup{dist } (\Omega_{\kappa}, \Omega_{\kappa'})>M^{\varepsilon_{2}} \textup{ for } \kappa\ne\kappa'.
    \end{align}
    Write $T=D+S$ ($D$ is a diagonal matrix) where 
    \begin{equation}
        \lVert S \rVert<\varepsilon,\ \  |S(\xi,\xi')|<\varepsilon e^{-|\xi-\xi'|^{c}},  
    \end{equation}
    where $c$ is sufficiently small  and
    \begin{align}
        |D(\xi)|>\rho\gg \varepsilon \textup{\ \  if\ \  } \xi\in\Omega\setminus \cup \Omega_{\kappa},\\
        \lVert (T|_{\tilde{\Omega}_{\kappa}})^{-1} \rVert<M^{C} \textup{ for all } \kappa,
    \end{align}
    where $\tilde{\Omega}_{\kappa}$ is an $M^{\varepsilon_{3}}$-neighborhood of $\Omega_{\kappa}$.
    Then 
    \begin{equation}
        \lVert (T|\Omega)^{-1} \rVert<\rho^{-1} M^{C+1},
    \end{equation}
    and \begin{equation}
        |(T|\Omega)^{-1}(\xi, \xi')|<e^{-\frac{1}{10}|\xi-\xi'|^{c}} \textup{\ \  if\ \  } |\xi-\xi'|>M^{2\varepsilon_{1}}.
    \end{equation}
\label{coupling lemma2}
\end{lemma}

\end{document}